\documentclass[a4paper,11pt]{article}

\usepackage[margin=1in]{geometry}

\usepackage[utf8x]{inputenc}
\usepackage{amsmath}
\usepackage{amssymb}
\usepackage{amsthm}
\usepackage{algorithmic}
\usepackage{graphicx}
\usepackage{algorithm}
\usepackage{color}

\newmuskip\pFqmuskip

\newcommand*\pFq[6][8]{%
  \begingroup 
  \pFqmuskip=#1mu\relax
  \mathcode`\,=\string"8000
  \begingroup\lccode`\~=`\,
  \lowercase{\endgroup\let~}\pFqcomma
  {}_{#2}F_{#3}{\left[\genfrac..{0pt}{}{#4}{#5};#6\right]}%
  \endgroup
}
\newcommand{\pFqcomma}{\mskip\pFqmuskip}

\newcommand{\SoS}{\text{\sc{SoS}}}

\newcommand{\set}[1]{\left\{ #1 \right\}}
\newcommand{\PS}{\mathcal{P}}

\newcommand{\y}{y^N}
\newcommand{\fallfac}[2]{{#1}^{\underline{#2}}}
\newcommand{\risefac}[2]{{#1}^{\overline{#2}}}

\newcommand{\R}{\mathbb{R}}

\usepackage{mathtools}

\newtheorem{theorem}{Theorem}
\newtheorem{lemma}[theorem]{Lemma}

\newtheorem{corollary}[theorem]{Corollary}
\newtheorem{definition}{Definition}

\title{Tight Sum-of-Squares lower bounds for binary polynomial optimization problems\footnote{Supported by the Swiss National Science Foundation project 200020-144491/1 ``Approximation Algorithms for Machine Scheduling Through Theory and Experiments''.}}

\author{ Adam Kurpisz \quad Samuli Lepp\"anen~\quad Monaldo Mastrolilli\\
\small{IDSIA, 6928 Manno,
Switzerland}
{\{adam, samuli, monaldo\}@idsia.ch}
}

\date{}

\begin{document}

\maketitle

\begin{abstract}
We give two results concerning the power of the Sum-of-Squares(SoS)/Lasserre hierarchy.
For binary polynomial optimization problems of degree $2d$ and an odd number of variables $n$,  we prove that $\frac{n+2d-1}{2}$ levels of the SoS/Lasserre hierarchy are necessary to provide the exact optimal value. This matches the recent upper bound result by Sakaue, Takeda, Kim and Ito.

Additionally, we study a conjecture by Laurent, who considered the linear representation of a set with no integral points. She showed that the Sherali-Adams hierarchy requires $n$ levels to detect the empty integer hull, and conjectured that the SoS/Lasserre rank for the same problem is $n-1$. We disprove this conjecture and derive lower and upper bounds for the rank.

\end{abstract}

\section{Introduction}

In this paper we are concerned with the unconstrained \emph{binary polynomial optimization problems} (BPOP):
$$\min_{x \in \set{0,1}^n} f(x)$$
where $f(x)$ is a multivariate polynomial.
Many basic optimization problems are special cases of this general problem. Prominent examples include the \textsc{MaxCut} problem and the boolean \textsc{Max $k$-csp}. For these problems the polynomials have at most degree 2 and $k$, respectively.

The Sum-of-Squares (SoS)/Lasserre hierarchy of semidefinite (SDP) relaxations~\cite{Lasserre01,parrilo00} is one of the most studied solution methods for general polynomial optimization problems (POP) including BPOP.
The hierarchy is parameterized by a parameter $t$ called the \emph{relaxation level} and larger levels correspond to tighter relaxations. At level $t$,
the relaxation consists of $n^{O(t)}$ variables and constraints, and it is thus solvable in time $n^{O(t)}$ using for example the ellipsoid method.
At level $n$ the SOS hierarchy finds the exact optimal value of an arbitrary constrained BPOP (but not a general POP). 

For \emph{quadratic} BPOP, Laurent~\cite{Laurent03a} conjectured  that at level $\lceil \frac{n}{2} \rceil$ the relaxation provides the exact optimal value.
She also provided a matching lower bound showing that $\lfloor \frac{n}{2} \rfloor$ levels are not enough for finding the integer cut polytope of the complete graph with $n$ nodes, when $n$ is odd (the result was preceded by a similar lower bound by Grigoriev~\cite{Grigoriev01} for the \textsc{Knapsack} problem).
 The conjecture was proved by Fawzi, Saunderson and Parrilo~\cite{FawziSaundersonParrilo15} while showing that $\lceil \frac{n}{2} \rceil$ rounds are enough to exactly solve \emph{any} unconstrained BPOP of degree~2. Very recently, Sakaue, Takeda, Kim and Ito~\cite{SakaueTKI16} extended the result of~\cite{FawziSaundersonParrilo15} and showed that the SoS hierarchy requires at most $\lceil (n+r-1)/2 \rceil$ rounds to find the exact optimal value of an unconstrained BPOP of degree $r$ with $n$ variables. Note that the two upper bounds~\cite{FawziSaundersonParrilo15,SakaueTKI16} coincide when $n$ is odd and $r = 2$, whereas for even $n$ there is a difference of 1 (although~\cite{SakaueTKI16} show also that if the optimized polynomial consists of only even degree monomials, the bound reduces to $\lceil (n+r-2)/2 \rceil$, matching the bound of~\cite{FawziSaundersonParrilo15} for example for the \textsc{MaxCut} problem for every $n$). Furthermore, Sakaue et al.~\cite{SakaueTKI16} numerically confirmed that for some degrees their bound is tight for certain instances of unconstrained BPOPs with~8 variables.

In a recent breakthrough Lee, Raghavendra and Steurer~\cite{LeeRagSteu15} proved that for the class of \textsc{Max-CSP}s the SoS relaxation yields the ``optimal'' SDP approximation, meaning that SDPs of polynomial-size are equivalent in power to those
arising from $O(1)$ rounds of the SoS relaxations.
This result implies that known lower bound for SoS SDP relaxations translates to corresponding lower bounds on the size of any SDP formulations. With this aim, they build on the work of Grigoriev/Laurent~\cite{Grigoriev01,Laurent03a} to show that, for odd $n$, any sum of squares of degree $\lfloor n/2\rfloor$ polynomials has $\ell_1$-error at least $2^{n-2}/\sqrt{n}$ in approximating the following quadratic function
 \begin{equation}\label{eq:pol2}
f(x)=(\|x\|_1-\lfloor n/2\rfloor)(\|x\|_1-\lfloor n/2\rfloor-1)
\end{equation}
This result is shown to imply lower bounds on the semidefinite extension complexity of the \emph{correlation polytope} (which is isomorphic to the cut polytope and sometimes also called
boolean quadric polytope). By reduction, the latter in turn implies exponential lower bounds for the integer cut, TSP and stable set polytopes. In~\cite{LeePWY16} Lee, Prakash, de Wolf and Yuen proved that these lower bounds cannot be improved by showing better $\ell_1$-approximations of $f(x)$.
%
%

\paragraph*{Our Results.}
In this paper we give two results concerning the power of the SoS hierarchy. Our first result shows that the bound given by Sakaue et al.~\cite{SakaueTKI16} is tight for polynomials with even degree and an odd number of variables. More precisely, we consider BPOPs of the form $\min_{x\in\{0,1\}^n} f_d(x)$ where $f_d(x)$ is a degree $2d$ (for $d\geq1$) polynomial defined as follows:
\begin{align}\label{eq:pold}
f_{d}(x)=  \fallfac{\left(\|x\|_1 - \lfloor n/2\rfloor + d -1 \right)}{2d}
\end{align}
where $\fallfac{k}{r} = k(k-1)\cdots(k-r+1)$ denotes the falling factorial. For $d=1$ we have $f_1(x)=f(x)$, where $f(x)$ is the polynomial defined in \eqref{eq:pol2} and considered in~\cite{LeeRagSteu15,LeePWY16}.
We show that for odd $n = 2m + 1$, the SoS relaxation allows negative values for polynomial $f_d(x)$ that is non-negative over $\set{0,1}^n$, even at level $\lceil \frac{n+2d-1}{2}\rceil - 1= m+d-1$.
%


Our second result concerns comparing the SoS hierarchy to other lift and project methods. A commonly used benchmark for comparing hierarchies is to find the smallest level at which they find the convex hull of a given set of integral points $P$,\footnote{The smallest such level is called the \emph{rank} of $P$, and it is always smaller or equal to $n$ for the usually studied hierarchies.} usually given as an intersection of the set $\set{0,1}^n$ and a polytope. Examples of such results include~\cite{GoemansT01,Grigoriev01,Grigoriev01b,GrigorievV01,Laurent03a,MathieuS09,StephenT99}. In~\cite{Laurent03}, Laurent shows that the Sherali-Adams hierarchy detects that the set
\begin{equation} \label{eq:Laurent_empty_set}
K = \set{0,1}^n \cap \set{x \in [0,1]^n ~|~ \sum_{r \in R} x_r + \sum_{r \in R \setminus N} (1-x_r) \geq \frac{1}{2} ~ \text{for all } R\subseteq N}
\end{equation}
is empty only after $n$ levels. She then conjectures the SoS rank of $K$ is $n-1$.
The polytope $K$ has been used earlier to show that $n$ iterations are needed also for the following procedures: the Lov\'asz-Schrijver $N_+$ operator (with positive semidefiniteness) \cite{GoemansT01}, the Lov\'asz-Schrijver $N_+$ operator combined with taking Chv\'atal cuts~\cite{cook2001matrix}, and the $N_+$ operator combined with taking Gomory mixed integer cuts (equivalent to disjunctive cuts) \cite{cornuejols2001rank}. In this paper we disprove Laurent's conjecture, and show that indeed the SoS rank of $K$ is bounded between $\Omega(\sqrt{n})$ and $n-\Omega(n^{1/3})$.

Interestingly, Au~\cite{Au14} and the authors of this paper~\cite{KurpiszLM15} independently considered the rank of a variation of the set $K$ where on the right hand side of the inequalities there is an exponentially small constant instead of $\frac{1}{2}$. Both works show that the rank of the modified $K$ is exactly $n$.

In our proofs we demonstrate the use of a recent theorem of the authors~\cite{KurpiszLM16} that simplifies the positive semidefiniteness (PSD) condition of the SoS hierarchy when the problem formulation is highly symmetric (as noted in \cite{LeePWY16}, Blekherman~\cite{Blekherman15} has also obtained a similar result that is still in preparation). 
Our first result is obtained by showing that a certain conical combination of solutions with non-integral relaxation value to the SoS relaxation for the function \eqref{eq:pol2} gives a negative SoS relaxation value for the polynomials \eqref{eq:pold} of degree $2d$. Then, for the first and the second result, we apply the theorem in~\cite{KurpiszLM16} to reduce the PSDness condition into showing that a particular inequality is satisfied for every polynomial with a certain form. Showing that the inequality is satisfied (lower bounds) or cannot be satisfied (upper bounds) then boils down to evaluating or approximating a certain combinatorial sum. 
Our results also answer the question in \cite{LeePWY16} regarding the applications of the theorem of~\cite{Blekherman15,KurpiszLM16}.

\section{The Sum-of-Squares hierarchy}

In this paper we consider the SoS hierarchy when applied to (i) unconstrained $0/1$ polynomial optimization problems, and (ii) approximating the convex hull of the set 
\begin{equation} \label{eq:definition_of_set_K}
P = \set{x \in \set{0,1}^n~|~g_{\ell}(x)\geq 0, \forall \ell\in [p]}
\end{equation}
where $g_\ell(x)$ are linear constraints and $p$ a positive integer. The form of the SoS hierarchy we use in this paper is equivalent to the one used in literature (see e.g.~\cite{BarakBHKSZ12,Lasserre01,Laurent03}) and follows from applying a change of basis to the dual certificate of the refutation of the proof system (see~\cite{KurpiszLM16} for the details on the change of basis and \cite{MekaPW15} for discussion on the connection to the proof system). We use this change of basis  in order to obtain a useful decomposition of the moment matrices as a sum of rank one matrices of special kind. 

For any $I\subseteq N=\{1,\ldots,n\}$, let $x_I$ denote the $0/1$ solution obtained by setting $x_i = 1$ for $i \in I$, and $x_i = 0$ for $i\in N\setminus I$. For a function $f:\set{0,1}^n \rightarrow \R$, we denote by $f(x_I)$ the value of the function evaluated at $x_I$. In the SoS hierarchy defined below there is a variable $\y_I$ that can be interpreted as the ``relaxed'' indicator variable for the solution $x_I$. We point out that in this formulation of the hierarchy the number of variables $\{\y_I:I\subseteq N\}$ is exponential in $n$, but this is not a problem in our context since we are interested in proving lower and upper bounds rather than solving an optimization problem.

Let $\PS_t(N)$ be the collection of subsets of $N$ of size at most $t\in\mathbb{N}$. For every $I \subseteq N$, the $q$-\text{zeta vector}  $Z_I \in \mathbb{R}^{\PS_q(N)}$ is a $0/1$ vector with $J$-th entry ($|J| \leq q$) equal to $1$ if and only if $J \subseteq I$.\footnote{In order to keep the notation simple, we do not emphasize the parameter $q$ as the dimension of the vectors should be clear from the context.}
Note that $Z_IZ_I^\top$ is a rank one matrix and the matrices considered in Definitions~\ref{def:sos_definition_polynomials} and~\ref{def:sos_definition_integer_hull} are linear combinations of these rank one matrices.

To simplify the presentation we define the SoS hierarchy separately for polynomial optimization problems and for the integer hull approximation.

\begin{definition}~\label{def:sos_definition_polynomials}
The $t$-th round SoS hierarchy relaxation of $\min_{x \in \set{0,1}^n} f(x)$, denoted by $\SoS_t(f)$, is the optimization problem with variables $\{\y_I \in \mathbb{R}:  \forall I \subseteq N\}$ of the form 
\begin{eqnarray}
 \min_{\y \in \mathbb{R}^{2^n}} \sum_{I \subseteq N} \y_I f(x_I)   & &  \label{eq:sos_pol_objective} \\
 \text{s.t. }  \sum_{\substack{I \subseteq N}} \y_I&= &1, \label{eq:sos_sum_condition_pol} \\
 \sum_{\substack{I \subseteq N}} \y_I Z_I Z_I^\top &\succeq& 0, \text{ where } Z_I \in \mathbb{R}^{\PS_{t}(N)} \label{eq:sos_variables_pol}
\end{eqnarray}
\end{definition}

\begin{definition}~\label{def:sos_definition_integer_hull}
 The $t$-th round SoS hierarchy relaxation for the set $P$ as given in~\eqref{eq:definition_of_set_K}, denoted by $\SoS_t(P)$, is the set of variables $\{\y_I \in \mathbb{R}:  \forall I \subseteq N\}$ that satisfy
 \begin{eqnarray}
 \sum_{\substack{I \subseteq N}} \y_I&=& 1, \label{eq:sos_sum_condition_ih} \\
 \sum_{\substack{I \subseteq N}} \y_I Z_I Z_I^\top &\succeq& 0, \text{ where } Z_I \in \mathbb{R}^{\PS_{t+1}(N)} \label{eq:sos_variables_ih}\\
 \sum_{\substack{I \subseteq N}} g_\ell(x_I)\y_I Z_IZ_I^\top &\succeq& 0, ~ \forall \ell\in [p]  \text{, where }  Z_I \in \mathbb{R}^{\PS_{t}(N)}  \label{eq:sos_constraints_ih}
\end{eqnarray}
\end{definition}
It is straightforward to see that the SoS hierarchy formulation given in Definition~\ref{def:sos_definition_integer_hull} is a relaxation of the integral polytope. Indeed consider any feasible integral solution $x_I \in P$ and set $\y_I=1$ and the other variables to zero. This solution clearly satisfies~\eqref{eq:sos_sum_condition_ih} and~\eqref{eq:sos_variables_ih} because the rank one matrix $Z_IZ_I^\top$ is positive semidefinite (PSD), and~\eqref{eq:sos_constraints_ih} since $x_I\in P$.

For a set $Q \subseteq [0,1]^n$, we define the projection from $\SoS_t(Q)$ to $\mathbb{R}^n$ as $x_i = \sum_{i \in I \subseteq N} y_I^N$ for each $i \in \set{1,...,n}$. The \emph{SoS rank} of $Q$, $\rho(Q)$, is the smallest $t$ such that $\SoS_t(Q)$ projects exactly to the convex hull of $Q \cap \set{0,1}^n$.

\subsection{Using symmetry to simplify the PSDness conditions} \label{sect:main_theorem}

In this section we present a theorem given in~\cite{KurpiszLM16} that can be used to simplify the PSDness conditions~\eqref{eq:sos_variables_pol},~\eqref{eq:sos_variables_ih} and~\eqref{eq:sos_constraints_ih} when the problem formulation is very symmetric. More precisely, the theorem can be applied whenever the solutions and constraints are symmetric in the sense that $w_I^N = w^N_J$ whenever $|I|=|J|$ where $w_I^N$ is understood to denote either $y^N_I$ or $g_\ell(x_I)y^N_I$. In what follows we denote by $\mathbb{R}[x]$ the ring of polynomials with real coefficients and by $\mathbb{R}[x]_d$ the polynomials in $\mathbb{R}[x]$ with degree less or equal to $d$.

\begin{theorem}[\cite{KurpiszLM16}] \label{thm:PSD_as_polynomial}
For any $t\in \{1,\ldots,n\}$, let $\mathcal{S}_t$ be the set of univariate polynomials $G_h(k)\in \R[k]$, for $h \in \{0,\ldots,t\}$, that satisfy the following conditions: 
\begin{align}
G_h(k)&\in \R[k]_{2t} \label{eq:degree}\\
G_h(k)&=0 \qquad \text{for }k\in \{0,\ldots,h-1\}\cup \{n-h+1,\ldots,n\} \text{, when }h\geq 1\label{eq:zeros}\\
G_h(k)&\geq0 \qquad \text{for }k\in [h-1, n-h+1]\label{eq:geq0}
\end{align} 
For any fixed set of values $\{w^N_k \in \R:k =0,
\ldots,n\}$,
if the following holds 
\begin{align}
\sum_{k=h}^{n-h}\binom{n}{k} w^N_k G_h(k) &\geq 0 \qquad \forall G_h(k)\in \mathcal{S}_t \label{eq:sym_psd_cond}
\end{align}
then 
$$
 \sum_{k = 0}^n w^N_k \sum_{\substack{I \subseteq N \\ |I| = k}} Z_IZ_I^\top \succeq 0 \qquad 
$$
where $Z_I \in \mathbb{R}^{\PS_{t}(N)}$.
\end{theorem}
Note that polynomial $G_h(k)$ in~\eqref{eq:geq0} is nonnegative in a \emph{real interval}, and in~\eqref{eq:zeros} it is zero over a \emph{set of integers}. Moreover, constraints \eqref{eq:sym_psd_cond} are trivially satisfied for $h> \lfloor n/2 \rfloor$.

\section{Tightness of the SoS upper bounds for unconstrained BPOPs}

In~\cite{SakaueTKI16} it is shown that the SoS hierarchy exactly solves any unconstrained BPOP of degree $r$ with $n$ variables after $\lceil \frac{n+r-1}{2}\rceil$ levels. We show that this bound is tight for certain values of $n$ and $r$, by giving a polynomial of degree $r=2d$ for $d \geq 1$ that is non-negative over the hypercube, and show that when $n = 2m+1$, $m\geq d$, the SoS relaxation of the corresponding BPOP attains a negative value at level $t=\lceil \frac{n+2d-1}{2}\rceil-1=m+d-1$. 

More precisely, we consider the degree $2d$ polynomial 
\begin{equation}
\label{eq:poly_2d}
f_d(x)=  \fallfac{\left(\|x\|_1 + d - m -1 \right)}{2d}
\end{equation}
where $\fallfac{k}{r} = k(k-1)\cdots(k-r+1)$ denotes the falling factorial and $\|x\|_1 = \sum_i x_i$. For the sake of convenience, we denote by $f_d(k)$ the univariate polynomial evaluated at any point $x$ with $\sum_i x_i = k$. We obtain the following result

\begin{theorem}
\label{thm:rank_of_poly_2d}
For odd $n$, the SoS relaxation of minimizing $f_{d}$ requires at least $\lceil \frac{n+2d-1}{2}\rceil$ levels to find the exact optimum.
\end{theorem}

\subsection{Proof of Theorem~\ref{thm:rank_of_poly_2d}}
\paragraph*{The case $d=1$.} The polynomial $f_1(x)$ is connected to the $\textsc{MaxCut}$ problem in the complete graph of $n=2m+1$ vertices in the following way: Let $x \in \set{0,1}^n$ denote any partition of the vertices into two sets in the natural way. Then, the maximal cut is achieved whenever $\sum_i x_i$ is either $m$ or $m+1$, and $m(m+1) - f_1(x)$ counts the edges in the cut. Therefore, the SoS hierarchy is not able to exactly solve the $\textsc{MaxCut}$ problem if it allows for solutions with negative values of the objective function~\eqref{eq:sos_pol_objective}.

It is shown in~\cite{KurpiszLM16} that\footnote{The same solution was earlier considered in different basis by~\cite{Grigoriev01b,Laurent03a} for the \textsc{Knapsack} and \textsc{MaxCut} problems respectively to show that the SoS hierarchy does not exactly solve the aforementioned problems at level $\lfloor \frac{n}{2}\rfloor$.} 
\begin{equation} \label{eq:Max_Cut_solution}
y_{I}^N[\alpha]=(n+1)\binom{\alpha}{n+1}\frac{(-1)^{n-|I|}}{\alpha-|I|} \qquad \forall I \subseteq N
\end{equation}
is a feasible solution to the SoS hierarchy (as given in Definition~\ref{def:sos_definition_polynomials}) at level $\lfloor \alpha \rfloor$ for any non-integer $0 < \alpha < \frac{n}{2}$. Since the value of the solution only depends on the size of the set $I$, we denote by $y_k^N[\alpha]$ any $y_I^N[\alpha]$ with $|I|=k$. As a consequence of the proof in~\cite{KurpiszLM16} it follows that for any non-integer $0 < \alpha \leq n$, $\sum_{i=0}^n \binom{n}{k}y_k^N[\alpha]=1$. Furthermore, it is shown that the objective function attains the value $\sum_{k=0}^n \binom{n}{k}y_k^N[\alpha] f_1(k)=f_1(\alpha)$ and that in particular for $\alpha = \frac{n}{2}$,  $f_1(\alpha) = -\frac{1}{4}$ at level $t=m$. Next we generalize this approach to $f_d(x)$. 

\paragraph*{Polynomials of degree 2d.} Consider the following solution
\begin{equation}
\label{eq:polynomial_solution}
z_k^N=(2d-2)!(n+1)\binom{\frac{n}{2}-d+1}{n+1}\frac{(-1)^{n-k}}{(\frac{n}{2}+d-1-k)^{\underline{2d-1}}} \qquad \forall k \in \{0,\ldots,n\}
\end{equation}
We show that for this solution, the SoS hierarchy objective~\eqref{eq:sos_pol_objective} attains a negative value (see Lemma~\ref{lem:objective_negative}) and~\eqref{eq:sos_variables_pol} is satisfied. For convenience, we do not actually show that~\eqref{eq:sos_sum_condition_pol} is satisfied and in fact it is not. We show, however, that $\sum_{k=0}^n \binom{n}{k}z_k > 0$, which implies that with proper normalization also~\eqref{eq:sos_sum_condition_pol} can be satisfied (see Lemma~\ref{lemma:positivity_of_z}). 

First we prove that the solution $z_k^N$ can be written as a conical combination of the solutions $y_k^N[\cdot]$ in~\eqref{eq:Max_Cut_solution}. We begin with the following lemma about partial fraction decompositions.
\begin{lemma}
\label{lem:partial_fraction_decomposition}
For any $b \in \mathbb{N}_+$ and $a \in \mathbb{R}$ the following identity holds
$$\frac{1}{\left(x-a\right)^{\underline{b}} }=\sum_{i=0}^{b-1}\frac{(-1)^{b-1-i}}{i!(b-1-i)!} \frac{1}{(x-a-i)}$$
\end{lemma}
\begin{proof}
It is known that given two polynomials $P(x)$ and $Q(x) = (x-a_1)(x-a_2) \cdots (x-a_n)$, where the $a_i$ are distinct constants and deg $P < n$, the rational polynomial $\frac{P(x)}{Q(x)}$ can be decomposed into
$$\frac{P(x)}{Q(x)} = \sum_{i=1}^n \frac{P(a_i)}{Q'(a_i)}\frac{1}{(x-a_i)}$$
where $Q'(x)$ is the derivative of $Q(x)$. In our case, since $P(x)=1$ and $Q(x)=\prod_{i=0}^{b-1} (x-a-i)$, we get
$$
\frac{1}{\left(x-a\right)^{\underline{b}} }=\sum_{i=0}^{b-1}\frac{1}{\prod_{j\neq i}(a+i-(a+j))} \frac{1}{(x-a-i)} =\sum_{i=0}^{b-1}\frac{(-1)^{b-1-i}}{i!(b-1-i)!} \frac{1}{(x-a-i)}
$$
\end{proof}

Now we can express the solution $z_k^N$ as a conical combination of the solutions $y_k^N[\cdot]$.
\begin{lemma}
\label{lem:decomposition_to_Max_Cut_solutions}
The solution~\eqref{eq:polynomial_solution} can be decomposed as a conical combination of $y_k^N[\cdot]$:
$$
z_k^N=\sum_{j=0}^{2d-2}  a_j  y_k^N[n/2+d-1-j] \qquad \forall k \in \{0,\ldots,n\}
$$
for positive
$$
a_j=\binom{2d-2}{j} \frac{(\frac{n}{2}+d-1)^{\underline{j}}}{(\frac{n}{2}-d+1+j)^{\underline{j}}}
$$
\end{lemma}
\begin{proof}
By Lemma~\ref{lem:partial_fraction_decomposition} we get that 
$$
\frac{1}{(\frac{n}{2}+d-1-k)^{\underline{2d-1}}}=\sum_{j=0}^{2d-2}\frac{(-1)^{2d-2-j}}{j!(2d-2-j)!} \cdot \frac{1}{(\frac{n}{2}+d-1-k-j)}
$$
and by writing
$$
\binom{\frac{n}{2}-d+1}{n+1}=     \frac{(-\frac{n}{2}+d-2-j)^{\underline{2d-2-j}}}{(\frac{n}{2}+d-1-j)^{\underline{2d-2-j}}}     \binom{\frac{n}{2}+d-1-j}{n+1}
$$
and using raising factorial notation, $\fallfac{(-b)}{a}=(-1)^a b^{\overline{a}}$, we get that
\begin{align*}
z_k^N&=\sum_{j=0}^{2d-2}\frac{(2d-2)!}{j!(2d-2-j)!}   \frac{(-\frac{n}{2}+d-2-j)^{\underline{2d-2-j}}}{(\frac{n}{2}+d-1-j)^{\underline{2d-2-j}}}(n+1)\binom{\frac{n}{2}+d-1-j}{n+1}  \cdot \frac{(-1)^{2d-2-j+n-k}}{(\frac{n}{2}+d-1-k-j)}\\
&=\sum_{j=0}^{2d-2} \binom{2d-2}{j}       \frac{(\frac{n}{2}-d+2+j)^{\overline{2d-2-j}}}{(\frac{n}{2}+d-1-j)^{\underline{2d-2-j}}}   y_k^N\left[\frac{n}{2}+d-1-j\right]\\
&=\sum_{j=0}^{2d-2} \binom{2d-2}{j}       \frac{(\frac{n}{2}+d-1)^{\underline{j}}}{(\frac{n}{2}-d+1+j)^{\underline{j}}}   y_k^N\left[\frac{n}{2}+d-1-j\right]
\end{align*}
\end{proof}

\begin{lemma} \label{lemma:positivity_of_z}
We have $\sum_{k=0}^n \binom{n}{k} z_k^N > 0$ for every odd $n$, $n=2m+1$, and $d\in [m]$.
\end{lemma}
\begin{proof}
The proof follows by recalling that for every $\alpha \in [0,n]\setminus \mathbb{Z}$, $\sum_{i=0}^n \binom{n}{k}y_k^N[\alpha]=1$ and by the fact that all the coefficients in the decomposition in Lemma~\ref{lem:decomposition_to_Max_Cut_solutions} are positive. 
\end{proof}

Now we show that the solution~\eqref{eq:polynomial_solution} is a feasible solution for the SoS hierarchy at level $t=m+d-1$. The solution~\eqref{eq:polynomial_solution} is symmetric, and so by Theorem~\ref{thm:PSD_as_polynomial} (see~\eqref{eq:sym_psd_cond}) is enough to prove that for $t=m+d-1$,
$$
\sum_{k=0}^n \binom{n}{k} z_k^N G_h(k) \geq 0\qquad \forall G_h(k) \in \mathcal{S}_t
$$

We first note that the solution~\eqref{eq:polynomial_solution} attains positive values for every integer $k \in \{m-d+1,\ldots, m+d\}$. Indeed, for $k = m-d+1+p$ for $p=\{0,\ldots,2d-1\}$, since 
$$\fallfac{\left(\frac{n}{2}-d+1\right)}{n+1}=\fallfac{\left(\frac{n}{2}-d+1\right)}{m-d+2}\left(\frac{1}{2}\right)^{\overline{m+d}}\left(-1\right)^{m+d}$$ 
and for $0\leq p \leq 2d-1$
$$\fallfac{\left(\frac{n}{2}+d-1-k\right)}{2d-1}=\fallfac{\left(2d-\frac{3}{2}-p\right)}{2d-1-p}\left(\frac{1}{2}\right)^{\overline{p}}\left(-1\right)^{p}$$ 
the only, not obviously, non-negative part of $z_k^N$ is
$$
\frac{(-1)^{m+d}(-1)^{m+d-p}}{(-1)^p}
$$
which is always positive.
 Thus the above (see~\eqref{eq:sym_psd_cond})  is always satisfied whenever $h \geq  m-d+1$ by the definition of the polynomials $G_h \in \mathcal{S}_t$.

It follows that it is enough to prove that the above is satisfied for $h \leq m-d$ which is implied if the following is true
$$
\sum_{k=0}^n \binom{n}{k} z_k^N P(k) \geq 0
$$
for every polynomial $P(x) \in \mathbb{R}[x]_{2t}$ that is nonnegative in the interval $[m-d+1, m+d]$.
\begin{lemma}
\label{lem:polynomial_main_thm}
For any polynomial $P(x) \in \mathbb{R}[x]_{2(m+d-1)} $ we have
$$
\sum_{k=0}^n \binom{n}{k} z_k^N P(k)= \sum_{j=0}^{2d-2} 
a_j
P\left(\frac{n}{2}+d-1-j\right)
$$
\end{lemma}
\begin{proof}

Let $g(k)=(\frac{n}{2}+d-1-k)^{\underline{2d-1}}$ be the polynomial of degree $2d-1$ that corresponds to the denominator in polynomial in $z_k^N$ (see~\eqref{eq:polynomial_solution}). By the polynomial remainder theorem, $P(k)=g(k)Q(k)+R(k)$, where the $Q(k)$ is the unique polynomial of degree at most $\deg(P)-\deg(g) \leq n-2$, and for the remainder it holds $R(r)=P(r)$ for all the roots $r$ of polynomial $g(k)$. Then
$$
\sum_{k=0}^n \binom{n}{k} z_{k}^N   P(k) = \sum_{k=0}^n \binom{n}{k} z_{k}^N g(k)Q(k)+\sum_{k=0}^n \binom{n}{k} z_{k}^N R(k)
$$
Here $\sum_{k=0}^n \binom{n}{k} z_{k}^N g(k)Q(k) = 0$, as $\sum_{k=0}^n (-1)^k\binom{n}{k} k^c=0$ for every $c\leq n-1$. We remark here that if the level $t$ is greater than $m+d-1$, then the polynomial $Q$ can be of degree $n$ or more and this reasoning fails. 

By Lemma~\ref{lem:decomposition_to_Max_Cut_solutions} we can write the sum with the remainder polynomial $R(k)$ as
$$
\sum_{j=0}^{2d-2} a_j \sum_{k=0}^n \binom{n}{k}  y_k^N\left[\frac{n}{2}+d-1-j\right] R(k) 
$$
and, again by the polynomial reminder theorem, for every $j\in \{0,\ldots,2d-2\}$, $R(k)=(\frac{n}{2}+d-1-j-k)S_j(k)+R(\frac{n}{2}+d-1-j)$ and as before, since the degree of $R$ is less or equal to $2d-2$, we have $\sum_{i=0}^n (-1)^k\binom{n}{k}S_j(k)=0$. Thus, since $R(r)=P(r)$ for all the roots $r$ of the polynomial $g(k)$, the above reduces to
\begin{align*}
&\sum_{j=0}^{2d-2} 
a_j
\left(  R\left(\frac{n}{2}+d-1-j\right)   \underbrace{  \sum_{k=0}^n \binom{n}{k}  y_k^N\left[\frac{n}{2}+d-1-j\right]}_{=1} \right)\\
=&\sum_{j=0}^{2d-2} 
a_j
P\left(\frac{n}{2}+d-1-j\right)
\end{align*}
\end{proof}

By Lemma~\ref{lem:polynomial_main_thm} we immediately obtain the following corollary.

\begin{corollary}
\label{cor:polynomial_main_thm}
For any polynomial $P(x) \in \mathbb{R}[x]_{2(m+d-1)} $ such that $P(x) \geq 0$ for $x \in [m-d+1, m+d]$ we have
$$
\sum_{k=0}^n \binom{n}{k} z_k^N P(k) \geq 0
$$
\end{corollary}
\begin{proof}
By Lemma~\ref{lem:polynomial_main_thm} we have that
$$
\sum_{k=0}^n \binom{n}{k} z_k^N P(k)= \sum_{j=0}^{2d-2} 
a_k
P\left(\frac{n}{2}+d-1-j\right)
$$
which is positive since it is a conical combination of points at which polynomial $P$ is positive.
\end{proof}

It remains to show that the objective value of the SoS hierarchy~\eqref{eq:sos_pol_objective} attains a negative value.

\begin{lemma}
\label{lem:objective_negative}
The sum $\sum_{k=0}^n \binom{n}{k} z_k^N f_d(k)$ is negative for every odd $n=2m+1$, for any positive integer $m$ and $d\in \set{1,...,m}$.
\end{lemma}
\begin{proof}
By Lemma~\ref{lem:polynomial_main_thm}, the solution $z_k^N$ is such that
$$
\sum_{k=0}^n \binom{n}{k} z_k^N f_d(k)= \sum_{j=0}^{2d-2} a_j f_d\left(\frac{n}{2}+d-1-j\right)
$$
Then, the claim is proved by showing that the following function $g(d,n)$ is negative for every odd $n=2m+1$, for any positive integer $m$ and $d\in[m]$. Formally, that
\begin{align*}
 g(d,n)=\sum_{j=0}^{2d-2} \binom{2d-2}{j}       \frac{\fallfac{(\frac{n}{2}+d-1)}{j}}{(\frac{n}{2}-d+1+j)^{\underline{j}}}      \fallfac{ \left(2d-3/2-j\right)}{2d} <0
\end{align*}
More precisely we show that the following identity holds (where $!!$ denotes the double factorial).
\begin{align}\label{eq:hg}
g(d,n)=  \fallfac{(2d-3/2)}{2d} \cdot \frac{4^{d-1} (2 d-2)! (2 d-1)!! }{(d-1)!  (4 d-3)!! }\cdot 
    \frac{(2 m - 2 d+3)!! (m-1)!  }{(m-d)! (2m+1)!!}
\end{align}
By simple inspection it is easy to see that \eqref{eq:hg} is negative and the claim follows.
 
We start by rewriting $g(d,n)$ by using the following (easy to check) identities:
 \begin{align*}
 \binom{2d-2}{j} \fallfac{\left(\frac{n}{2}+d-1\right)}{j} &= \frac{\risefac{(2-2d)}{j} \risefac{\left(1-d-\frac{n}{2}\right)}{j}}{j!}\\
 \fallfac{\left(\frac{n}{2}-d+1+j\right)}{j}&=\risefac{\left(2-d+\frac{n}{2}\right)}{j}\\
 \fallfac{ \left(2d-3/2-j\right)}{2d}&= \fallfac{(2d-3/2)}{2d} \frac{\risefac{(3/2)}{j}}{\risefac{(3/2-2d)}{j}}
 \end{align*}
By the above identities we have that
 \begin{align}
 g(d,n)&=\sum_{j=0}^{2d-2} \binom{2d-2}{j}       \frac{\fallfac{(\frac{n}{2}+d-1)}{j}}{(\frac{n}{2}-d+1+j)^{\underline{j}}}      \fallfac{ \left(2d-3/2-j\right)}{2d} \nonumber \\
 &= \fallfac{(2d-3/2)}{2d} \sum_{j=0}^{2d-2}    \frac{\risefac{(2-2d)}{j} \risefac{\left(1-d-\frac{n}{2}\right)}{j} \risefac{(3/2)}{j}}{\risefac{\left(2-d+\frac{n}{2}\right)}{j} \risefac{(3/2-2d)}{j}}\cdot \frac{1}{j!}\nonumber\\
 &= \fallfac{(2d-3/2)}{2d} \sum_{j=0}^{\infty}    \frac{\risefac{(2-2d)}{j} \risefac{\left(1-d-\frac{n}{2}\right)}{j} \risefac{(3/2)}{j}}{\risefac{\left(2-d+\frac{n}{2}\right)}{j} \risefac{(3/2-2d)}{j}}\cdot \frac{1}{j!}\nonumber\\
 &= \fallfac{(2d-3/2)}{2d} \cdot \pFq{3}{2}{a,b,c}{1+a-b,1+a-c}{1}\label{eq:gdn}
\end{align}
where $\pFq{3}{2}{a,b,c}{1+a-b,1+a-c}{1}=\sum_{j=0}^{\infty}    \frac{\risefac{(a)}{j} \risefac{(b)}{j} \risefac{( c )}{j}}{\risefac{(1+a-b)}{j} \risefac{(1+a-c)}{j}}\cdot \frac{1}{j!}$ is the generalized hypergeometric series with $a = 2 - 2 d$, $b = 1 - d - n/2$ and $c = 3/2$.

Note that $1+a/2-b-c = \frac{n-1}{2}>0$ and
by using Dixon's identity~\cite{dixon1902summation} for the generalized hypergeometric series $\pFq{3}{2}{a,b,c}{1+a-b,1+a-c}{1}$ (when $\Re(1+a/2-b-c)>0$) we have
\begin{align}\label{eq:dixon}
 \pFq{3}{2}{a,b,c}{1+a-b,1+a-c}{1} &=   \frac{\Gamma(1+\frac{1}{2}a)\Gamma(1+a-b)\Gamma(1+a-c)\Gamma(1+\frac{1}{2}a-b-c)}{\Gamma(1+a)\Gamma(1+\frac{1}{2}a-b)\Gamma(1+\frac{1}{2}a-c)\Gamma(1+a-b-c)}\nonumber\\
&= \frac{\Gamma(2 - d)\Gamma(5/2 - d + m)\Gamma(3/2 - 2 d)\Gamma(m)}{\Gamma(3 - 2 d)\Gamma(3/2 + m)\Gamma(1/2 - d)\Gamma(1 - d + m)}\nonumber
\end{align}
Note that $\Gamma(2 - d)=(1-d)\Gamma(1 - d)$ and $\Gamma(3 - 2 d)=2(1-d)(1-2d)\Gamma(1 - 2 d)$.
By using the Euler's reflection formula we have that $\frac{\sin{(\pi d)}}{\pi}=\frac{1}{\Gamma(1-d)\Gamma(d)}$ and $\frac{\sin{(\pi 2d)}}{\pi}=\frac{1}{\Gamma(1-2d)\Gamma(2d)}$, and by the integrality of $d$ we have that
 
$$\frac{\Gamma(1 - d)}{\Gamma(1 - 2 d)}=\frac{\sin{(\pi 2d)} (2d-1)!}{\sin{(\pi d)} (d-1)!}
 =\frac{2\cos{(\pi d)} (2d-1)!}{ (d-1)!}=\frac{2(-1)^{d} (2d-1)!}{ (d-1)!}$$
Therefore
 $$\frac{\Gamma(2 - d)}{\Gamma(3 - 2 d)}=(-1)^{d+1}\frac{ (2d-2)!}{(d-1)!}$$
Recall that for nonnegative integer values of $x$ we have $\Gamma(\frac{1}{2} - x)=\frac{(-2)^x}{(2n-1)!!}\sqrt{\pi}$ and $\Gamma(\frac{1}{2} + x)=\frac{(2x-1)!!}{2^x}\sqrt{\pi}$, and the following holds.
\begin{align*}
 \pFq{3}{2}{a,b,c}{1+a-b,1+a-c}{1} &=(-1)^{d+1}\frac{ (2d-2)!}{(d-1)!} \frac{(m-1)!}{(m-d)!}\cdot \frac{\Gamma(\frac{5}{2} - d + m)\Gamma(\frac{3}{2} - 2 d)}{\Gamma(\frac{3}{2} + m)\Gamma(\frac{1}{2} - d)}\nonumber\\
 &= \frac{4^{d-1} (2 d-2)! (2 d-1)!! }{(d-1)!  (4 d-3)!! }
     \frac{(2 m - 2 d+3)!! (m-1)!  }{(m-d)! (2m+1)!!}
\end{align*}
By simple inspection we see that $\pFq{3}{2}{a,b,c}{1+a-b,1+a-c}{1}$ is always positive and $g(d,n)$, see \eqref{eq:gdn}, is negative as claimed.
\end{proof}

\section{Rank bounds for detecting a particular empty integral hull}

In~\cite{Laurent03} Laurent considers the representation of the empty set as~\eqref{eq:Laurent_empty_set} and shows that the Sherali-Adams procedure requires $n$ levels to detect that $K = \emptyset$. She conjectures that the SoS rank of $K$ is $n-1$. In this section we disprove this conjecture and derive a lower and upper bound for the SoS rank of $K$.

\begin{theorem}\label{thm:rank_bounds_of_K}
The SoS rank of $K$ in~\eqref{eq:Laurent_empty_set} can be bounded by $\Omega(\sqrt{n}) \leq \rho(K) \leq n - \Omega(n^\frac{1}{3})$.
\end{theorem}

\begin{proof}\item
\paragraph*{The upper bound.} 

By symmetry, the solution $y_I^N = \frac{1}{2^n}$ for each $I \subseteq N$ is feasible to $\SoS_t(K)$ unless $\SoS_t(K) = \emptyset$. Let us assume that such a solution is feasible and consider the constraint of $K$ corresponding to $R = N$. Then, $g_R(x_I)$ is negative only when $I = \emptyset$.

To analyse the PSDness, we apply Theorem~\ref{thm:PSD_as_polynomial}. Notice that in this case we can assume that $P(k)$ is of the form $G_0(k)$, since if $h > 0$, the only negative term in the sum~\eqref{eq:sym_psd_cond} corresponding to $k=0$ is canceled due to $G_h(0)=0$, and the inequality holds trivially. Therefore, the PSDness condition~\eqref{eq:sos_constraints_ih} reduces to 
\begin{equation} \label{eq:Laurent_UB_PSD1}
\sum_{k=0}^n \binom{n}{k}\frac{1}{2^n}\left(k-\frac{1}{2}\right)P^2(k) \geq 0
\end{equation}
for every polynomial $P$ of degree $t$. Importantly, what is not mentioned in the statement of Theorem~\ref{thm:PSD_as_polynomial}, in this case the PSDness condition actually becomes an if and only if condition (see Theorem 7 in~\cite{KurpiszLM16}). Therefore, showing that~\eqref{eq:Laurent_UB_PSD1} is not satisfied implies that the PSDness condition~\eqref{eq:sos_constraints_ih} is not satisfied.

We now fix the polynomial as $P(k) = \prod_{i=1}^t(n-k-i+1)$, i.e., such that $P$ has the roots at $n, n-1, ..., n-t+1$, and argue that such a polynomial can never satisfy~\eqref{eq:Laurent_UB_PSD1} when $t$ is large. Indeed, rewriting the condition using this polynomial, removing the redundant factor $\frac{1}{2^n}$ and moving the negative term to the right hand side, we have the necessary requirement for the positive semidefiniteness that
$$
\sum_{k=1}^{n-t} \binom{n}{k}\left(k-\frac{1}{2}\right)\prod_{i=1}^t(n-k-i+1)^2 \geq \frac{1}{2}\prod_{i=1}^t(n-i+1)^2
$$
Notice that now the sum goes up to $n-t$ only, since all the terms $k > n-t$ are 0 by our choice of the polynomial. 
By dividing both sides by the positive term $\prod_{i=1}^t(n-i+1)^2$ and observing that $\frac{\prod_{i=1}^t(n-k-i+1)}{\prod_{i=1}^t(n-i+1)} = \frac{\fallfac{(n-t)}{k}}{\fallfac{n}{k}}$, the condition further simplifies to
\begin{equation} \label{eq:Laurent_UB_PSD2}
\sum_{k=1}^{n-t} \binom{n}{k}\left(k-\frac{1}{2}\right)\left(\frac{\fallfac{(n-t)}{k}}{\fallfac{n}{k}}\right)^2 \geq \frac{1}{2}
\end{equation}
Next we upper bound the sum on the left hand side of~\eqref{eq:Laurent_UB_PSD2} by considering a generic element for any $ 1 \leq k \leq n-t$. Any element can be bounded by
$$
\binom{n}{k}\left(k-\frac{1}{2}\right) \left( \frac{\fallfac{(n-t)}{k}}{\fallfac{n}{k}}\right)^2 \leq \frac{n^ke^k}{k^k}k\frac{(n-t)^{2k}}{(n-k)^{2k}} \leq  \frac{e^k}{k^{k-1}} \left(\frac{n(n-t)^2}{(n-k)^2}\right)^k
$$
Here we use $\frac{e^k}{k^{k-1}} < 3$ for any $k$ and $\frac{1}{n-k} \leq \frac{1}{t}$ for $k \leq n-t$. Then, for $t \geq n - o(\sqrt{n})$, it holds $\frac{n(n-t)^2}{t^2} < 1$ so we can approximate
$$
\frac{e^k}{k^{k-1}} \left(\frac{n(n-t)^2}{(n-k)^2}\right)^k \leq 3\frac{n(n-t)^2}{t^2}
$$

Now, the sum on the left hand side of~\eqref{eq:Laurent_UB_PSD2} is upper bounded by $3(n-t)\frac{n(n-t)^2}{t^2}$ and thus the solution is never feasible to $\SoS_t(K)$ if
$$
3\frac{n(n-t)^3}{t^2} < \frac{1}{2}
$$
Setting $t = n - C n^\frac{1}{3}$ satisfies the inequality asymptotically for an appropriate constant $C$.

\paragraph*{The lower bound.} We show that the symmetric solution $y^N_I = \frac{1}{2^n}$ is feasible for $\SoS_t(K)$ when $t$ is $\Omega(\sqrt{n})$. Again, by symmetry it is enough to show that the moment matrix of one constraint is PSD, and again we consider the constraint corresponding to $R = N$. Therefore, we need to show that~\eqref{eq:Laurent_UB_PSD1} is satisfied for any choice of the polynomial $P$ with degree less or equal to $t$. Writing the polynomial $P$ in root form with roots $r_i,~ i=1,...,t$, we get similarly as in~\eqref{eq:Laurent_UB_PSD2} the condition
\begin{equation}\label{eq:Laurent_LB_PSD1}
\sum_{k=1}^{n} \binom{n}{k}\left(k-\frac{1}{2}\right)\prod_{i=1}^t\left(\frac{k-r_i}{r_i}\right)^2 \geq \frac{1}{2}
\end{equation}
Now, we seek for a lower bound for the sum on the left hand side and find the condition on $t$ such that the lower bound still exceeds $\frac{1}{2}$. 

One can show (see~\cite{KurpiszLM16}) that the roots $r_i$ can be assumed to be real and to be located in the interval $[0,n]$. Furthermore, we can assume that the polynomial has degree of exactly $t$. Then, we look for the worst-case assignment for the roots. 

No matter how the roots are located, there exist at least one non-zero point $k \in N$ such that $|k-r_i| \geq \frac{n}{2(t+1)}$ for every root $r_i$ and $k \geq \lfloor \frac{n}{2(t+1)} \rfloor$. In the worst case the smallest of such points is $\lfloor\frac{n}{2(t+1)}\rfloor$. Let $u = \lfloor\frac{n}{2(t+1)}\rfloor$ be this point. Then, \eqref{eq:Laurent_LB_PSD1} is satisfied if we can show that
$$
\binom{n}{u}\left(u-\frac{1}{2}\right) \frac{\left(\frac{n}{2(t+1)}\right)^{2t}}{\prod_{i=1}^tr_i^2} \geq \frac{1}{2}
$$
Next, the worst case of the location for the roots in this formulation is $r_i = n$ for every $i=1,...,t$, since all the roots can be assumed to be less or equal to $n$. We can also get rid of the term $u-\frac{1}{2}$, since it is always greater than 1. We then obtain that~\eqref{eq:Laurent_LB_PSD1} holds if
$$
\binom{n}{u}\frac{\left(\frac{n}{2(t+1)}\right)^{2t}}{n^{2t}} \geq \frac{1}{2}
$$
Here we use the inequality $\binom{n}{u} > \frac{n^u}{u^u}$ and the fact that $\frac{1}{2(t+1)} \geq \frac{1}{4t}$ to get that if $\frac{n^u}{u^u}(4t)^{-2t} \geq \frac{1}{2}$
holds, then the solution is feasible for $\SoS_t(K)$. We have that $n \geq tu$, so the above is satisfied if
$$
\frac{(tu)^u}{u^u}(4t)^{-2t}  \geq \frac{1}{2} \Leftrightarrow t^u(4t)^{-2t} \geq \frac{1}{2}
$$
We have that $u \geq \frac{n}{4t}$, so it is enough to satisfy $t^{\frac{n}{4t}}(4t)^{-2t} \geq \frac{1}{2}$, which is equivalent to $4^{-2t}t^{\frac{n}{4t}-2t} \geq \frac{1}{2}$.
If here $t = \frac{\sqrt{n}}{4}$ we need to then satisfy $4^{-\sqrt{n}}\sqrt{n}^{\sqrt{n}/2} \geq \frac{1}{2}$, which holds asymptotically in $n$.
\end{proof}

\textbf{Open question.} We note that applying the theorem of~\cite{KurpiszLM16}, it is possible to perform numerical experiments to test the SoS rank of the polytope $K$ with large number of variables. For a fixed number of variables, the upper bound can be experimented by fixing the polynomial $P$ in~\eqref{eq:Laurent_UB_PSD1} in some systematic way and by finding the level $t$ (i.e., the number of roots) for which the expression is positive/negative for that polynomial. For the lower bound, the polynomial in~\eqref{eq:Laurent_UB_PSD1} can be expressed in the root form and the PSDness can be tested using a numerical solver to minimize the resulting expression, where the roots are the variables to be minimized.

Based on such experiments, we conjecture that the SoS rank of $K$ is ``close'' to $\frac{n}{2}$ and suggest that our bounds in Theorem~\ref{thm:rank_bounds_of_K} are far from being tight. Therefore we leave it as an open question to improve our bounds for the rank. 

\paragraph*{Acknowledgements.}
The authors would like to express their gratitude to Alessio Benavoli for helpful discussions.

{\small
\bibliographystyle{abbrv}
\bibliography{icalp2016}
}


\appendix
\section*{Appendix}

\section{Missing proofs}

\begin{lemma} \label{lemma:symmetric_solution_to_K}
For any $t \geq 1$, if $\SoS_t(K) \neq \emptyset$, then the solution $y_I^N = \frac{1}{2^n}$ for each $I \subseteq N$ is a feasible solution for $\SoS_t(K)$.
\end{lemma}
\begin{proof}

Let $d$ be either $t$ or $t+1$. First we observe that there is a one-to-one correspondence between vectors $v \in \mathbb{R}^{\PS_d(N)}$ and multilinear polynomials $p(x)$ of degree $d$ over $\set{0,1}^n$. Using the basis where $v_I$ is the coefficient of the monomial $\prod_{i \in I}x_i$ in $p(x)$, it can be easily seen that $v^\top Z_I = p(x_I)$. 

Assume that one of the matrices associated with the SoS hierarchy is PSD (where $z_I^N$ is either $y_I^N$ or $g_\ell(x_I)y_I^N$), in other words 
\begin{equation} \label{eq:symmetry_proof_PSD}
v^\top\left(\sum_{I \subseteq N} z_I^N Z_IZ_I^\top\right)v= \sum_{I \subseteq N} z_I^N (Z_I^\top v)^2 = \sum_{I \subseteq N} z_I^N p^2(x_I) \geq 0
\end{equation}
for every vector $v \in \mathbb{R}^d$ and thus for every multilinear polynomial $p$ of degree $d$. Then, consider the polynomial $q_S(x)$ ``rotated'' by the set $S \subseteq N$ defined as $q_S(x_I) = p(x_{I \bigtriangleup S})$ for every point $x_I$. In words, the polynomial $q_S$ evaluates the polynomial $p$ such that it replaces $x_i$ by $1-x_i$ if $i \in S$. Thus, $q_S$ and $p$ have the same degree. Then, 
$$
\sum_{I \subseteq N} z_I^N q^2(x_I) = \sum_{I \subseteq N} z_I^N p^2(x_{I\bigtriangleup S}) = \sum_{I \subseteq N} z_I^N (Z_{I \bigtriangleup S}^\top v)^2
$$ 
showing that the matrix $\sum_{I \subseteq N} z_I^N Z_{I \bigtriangleup S}Z_{I \bigtriangleup S}^\top$ is PSD if the matrix in~\eqref{eq:symmetry_proof_PSD} is PSD. 

For the constraints of the polytope $K$ as in~\eqref{eq:Laurent_empty_set} we have have $g_R(x_I) = |N \setminus (R \bigtriangleup I)| - \frac{1}{2}$. By this observation for any $S \subseteq N$ it holds that that $g_{R \bigtriangleup S}(x_{I \bigtriangleup S}) = g_R(x_I)$.
 
Now, assume that there exists a solution $y^N$ to $\SoS_t(K)$. Then, for any $S \subseteq N$ also the solution $u^N$ such that $u^N_I = y^N_{I \bigtriangleup S}$ must be feasible. This is because for the solution $u^N$ we can write the PSDness condition of the constraint $R$ as
$$
\sum_{I \subseteq N} u_I^N g_{R}(x_I) Z_IZ_I^\top = \sum_{I \subseteq N} y_I^N g_{R \bigtriangleup S}(x_I)Z_{I\bigtriangleup S}Z_{I\bigtriangleup S}^\top
$$
which must be PSD by assumption and the above discussion. Averaging over all $S$ yields a symmetric solution to $\SoS_t(K)$.
\end{proof}

\begin{lemma} \label{lemma:three_facts}
For the polynomial in~\eqref{eq:Laurent_LB_PSD1} we have
\begin{enumerate}
 \item[(a)] all the roots $r_1,..., r_t$ are real numbers,
 \item[(b)] all the roots $r_1,..., r_t$ are in the range, $1 \leq r_j \leq n$ for all $j = 1,\ldots,t$,
 \item[(c)] the degree of the polynomial is exactly $t$.
\end{enumerate} 
\end{lemma}

\begin{proof}
The proofs follow by inspecting~\eqref{eq:Laurent_LB_PSD1}.
\begin{enumerate}
 \item [(a)] Assume that some of the roots are complex and recall that complex roots of polynomials with real coefficients appear in conjugate pairs, i.e. $r_{2j-1} = a_j + b_ji$, $r_{2j} = a_j-b_ji$ for $j = 1,...,q$. Let $P'(k)$ be the polynomial with all real roots such that $r'_{2j-1} = r'_{2j} = \sqrt{a_{2j}^2 + b_{2j}^2}$ for $j=1,...,q$ and $r'_j = r_j, j>2q$.

For any $k\in N$ and $j\in[t]$, a simple calculation shows that
$$
\left(\frac{r_{2j-1}-k}{r_{2j-1}}\right)^2\left(\frac{r_{2j}-k}{r_{2j}}\right)^2 \geq \left(\frac{r'_{2j-1}-k}{r'_{2j-1}}\right)^2\left(\frac{r'_{2j}-k}{r'_{2j}}\right)^2
$$
Hence,
$$
\sum_{k = 1}^n  \binom{n}{k}\left(k-\frac{1}{2} \right)   \prod_{j=1}^t \left(\frac{r_j-k}{r_j}\right)^2 \geq
\sum_{k = 1}^n  \binom{n}{k} \left(k-\frac{1}{2} \right)  \prod_{j=1}^t \left(\frac{r'_j-k}{r'_j}\right)^2\\
$$

\item [(b)] Assume exactly one of the roots is negative, i.e., $r_1 = -a$, for $a>0$. Let $P'(k)$ be the univariate polynomial with all positive roots such that $r'_{1} = a$ and $r'_j = r_j, j>1$.
We have then for any $k\in N$ that
$$
 \left(\frac{-a-k}{-a}\right)^2\geq  \left(\frac{a-k}{a}\right)^2
$$
Similarly, let $P(k)$ be the univariate polynomial with $r_1 \in (0,1)$ and $r_j\geq 1$, for  $j>1$. Again, let $P'(k)$ be the univariate polynomial with $r_1=1$ and $r'_j = r_j, j>1$.
Again for any $k\in N$
$$
 \left(\frac{r_1-k}{r_1}\right)^2\geq  \left(\frac{1-k}{1}\right)^2
$$
Finally, let $P(k)$ be the univariate polynomial with $r_t =an$ for $a>1$ and $r_j \in [1,n]$, for  $j\neq t$. Let $P'(k)$ be the univariate polynomial with $r_t=n$ and $r'_j = r_j, j\neq t$.
As in the above cases, for any $k\in N$ we have
$$
 \left(\frac{an-k}{an}\right)^2\geq  \left(\frac{n-k}{n}\right)^2
$$

It follows that in each case it holds
$$
\sum_{k = 1}^n \binom{n}{k}\left(k-\frac{1}{2} \right)   \prod_{j=1}^t \left(\frac{r_j-k}{r_j}\right)^2 \geq
\sum_{k = 1}^n  \binom{n}{k} \left(k-\frac{1}{2} \right)  \prod_{j=1}^t \left(\frac{r'_j-k}{r'_j}\right)^2\\
$$

\item[(c)]  Let $P(k)$ be the univariate polynomial with degree $s<t$ with all real roots $r_j$. Let $P'(k)$ be the polynomial of degree $t$ with all real roots such that $r'_j = r_j, j\leq s$ and $r'_j=n$ for $s<j\leq t$. For any $k\in N$, we have

$$
1 \geq \left(\frac{n-k}{n}\right)^2
$$
Hence,
$$
\left(\frac{r_1-k}{r_1}\right)^2\cdots \left(\frac{r_{s}-k}{r_{s}}\right)^2 \geq \left(\frac{r_1-k}{r_1}\right)^2\cdots \left(\frac{r_{s}-k}{r_{s}}\right)^2 \left(\frac{n-k}{n}\right)^{2(t-s)}
$$
and finally
$$
\sum_{k = 1}^n \binom{n}{k}\left(k-\frac{1}{2} \right)  \prod_{j=1}^s \left(\frac{r_j-k}{r_j}\right)^2 \geq
\sum_{k = 1}^n   \binom{n}{k}\left(k-\frac{1}{2} \right)  \prod_{j=1}^t \left(\frac{r'_j-k}{r'_j}\right)^2\\
$$

\end{enumerate}
\end{proof}

\end{document}